\documentclass[11pt]{article}
\pdfoutput=1
\newif\ifFull
\Fulltrue
\usepackage{url}
\usepackage{graphicx}
\usepackage[noend]{algorithmic}
\ifFull
\usepackage{mathptmx}
\usepackage{compress}
\usepackage{fullpage}
\fi

\ifFull

\newtheorem{theorem}{Theorem}
\newtheorem{lemma}[theorem]{Lemma}

\newtheorem{corollary}[theorem]{Corollary}

\newcommand{\remark}[1]{\relax}
\newcommand{\qed}{{\hspace*{\fill}\rule{6pt}{6pt}}}
\newenvironment{proof}{\noindent{\bf Proof:}}{\smallskip}
\makeatletter
\def\@begintheorem#1#2{\sl \trivlist \item[\hskip \labelsep{\bf #1\ #2:}]}
\def\@opargbegintheorem#1#2#3{\sl \trivlist
      \item[\hskip \labelsep{\bf #1\ #2\ #3:}]}
\makeatother
\fi
\newcommand{\R}{{\bf R}}

\begin{document}

\title{Cloning Voronoi Diagrams via \\ Retroactive Data Structures}

\ifFull
\author{
Matthew T.~Dickerson \\
Dept.~of Math and Computer Sci. \\
Middlebury College 
\and
David Eppstein \\
Dept.~of Computer Science \\
Univ.~of California, Irvine 
\and
Michael T.~Goodrich \\
Dept.~of Computer Science \\
Univ.~of California, Irvine
}
\fi

\date{}
\maketitle

\begin{abstract}
We address the problem of 
replicating a Voronoi diagram $V(S)$ of a planar point set $S$
by making proximity queries\ifFull,  which are of three possible 
(in decreasing order of information content)\fi: 
\begin{enumerate}
\item
the exact location of the nearest site(s) in $S$
\item
the distance to and label(s) of the nearest site(s) in $S$
\item
a unique label for every nearest site in $S$.
\end{enumerate}
\ifFull
We provide algorithms showing how queries of 
Type 1 and Type 2 allow an exact cloning
of $V(S)$ with $O(n)$ queries and $O(n \log^2 n)$ processing time. 
We also prove that queries
of Type 3 can never \emph{exactly} 
clone $V(S)$, 
but we show that with $O(n \log\frac{1}{\epsilon})$
queries we can construct an $\epsilon$-approximate cloning of $V(S)$.
\fi
In addition to showing the limits of nearest-neighbor database
security, our methods also provide one of the first natural
algorithmic applications of retroactive data structures\ifFull, as several
of our methods critically rely on the use of such structures, one of
which is new to this paper\fi.
\end{abstract}

\ifFull
\thispagestyle{empty}
\clearpage
\fi

\section{Introduction} 

\ifFull
Data protection
is an important and growing concern in the information age.
Many Internet services derive income from online queries
of their databases; hence, they have an economic interest in keeping others
from replicating their data and offering a competing service.
That is, from a technological point of view, an online service's
economic foundation depends on preventing others from cloning
its data.
Unfortunately, each time such a service answers a query, it leaks a little
piece of its database.
Thus, the information security of such a service can be characterized
in terms of the existence of efficient algorithms that can exploit 
the data leakage present in each response to
systematically extract information from that service 
to be able to replicate that service.
\fi

In the \emph{algorithmic data-cloning framework}~\cite{g-magd-09},
a \emph{data querier}, Bob, is
allowed certain types of queries to a data set $S$ that belongs
to a \emph{data owner}, Alice.
Once Alice has determined the kinds of queries that she will
allow, she must correctly answer every valid query from Bob.
\ifFull
Bob has no other way to access $S$ except through these queries.
\fi
The information security question, then,
is to determine 
how many queries and how much processing time is needed for Bob to clone the entire data set.
We define a \emph{full cloning} of $S$ to mean
that Bob can
\ifFull 
replicate Alice's
API so that he can
\fi
answer any validly-formed 
query as accurately as Alice could.
In an \emph{$\epsilon$-approximate cloning} of $S$, 
Bob can answer any validly-formed
query to within an accuracy of $\epsilon>0$.
\ifFull
Thus, from the perspective of algorithmic research,
finding efficient
data-cloning algorithms for Bob demonstrates increased
information-security risks for
Alice, and showing algorithmic lower bounds for
Bob demonstrates improved security for Alice.
\fi

In this paper, we are interested 
in data sets consisting of a set $S$ of $n$ points in the plane, 
where $n$ and the contents of $S$ are initially unknown.
We study the risks to $S$ when Alice supports planar nearest-neighbor queries on $S$.
\ifFull
For example, Alice could represent the web site for 
a large coffee restaurant franchise, which provides the location of
its nearest coffee restaurant given the GPS coordinates of a
querier.
\fi
We assume that all the sites in $S$ are inside a known bounding box, 
$B$, which, without loss of generality, can be 
assumed to be a square with sides normalized to have length $1$. 
Since planar nearest-neighbor queries define a Voronoi diagram in the plane
(e.g., see~\cite{bkos-cgaa-97}), we can view Bob's goal in this instance of the
algorithmic data-cloning framework as that of trying to determine
the Voronoi diagram of $S$ inside the bounding box $B$.
\ifFull
If Alice's API allows Bob to exactly determine the Voronoi diagram of $S$,
then we assume Bob is interested in a full cloning of $S$.
Otherwise, if an exact determination of the set $S$ or 
a representation of the Voronoi diagram is impossible to determine, 
then Bob is interested in an $\epsilon$-approximate cloning of $S$, which, in
this case, would be a determination of 
the Voronoi diagram of $S$ to within a Euclidean distance of $\epsilon>0$. 
Naturally, Bob wishes to clone $S$ minimizing both the number of queries 
made and the running time to process the results of those queries. 

The difficulty of cloning Alice's database depends in part on the information she provides in response to a nearest-neighbor query point~$p$.
\fi
We consider three types of responses
(in decreasing order of information content): 
\begin{enumerate}
\item
the exact location of the nearest site(s) to $p$ in $S$
\item
the distance and label(s) of the nearest site(s) to $p$ in $S$
\item
a unique label identifying each nearest site to $p$ in $S$.
\end{enumerate}
\ifFull
Implicitly,
all three query types request information about a site in $S$ 
whose Voronoi region
the query point $p$ is contained within. 
Query 1, for example, asks for the exactly location and label of
site $q \in S$ whose Voronoi region $V(q)$ contains $p$, or if $p$'s  
nearest neighbor is 
not unique for the Voronoi edge or Voronoi vertex on which $p$ falls. 
\fi
With all three cases, we want to know how 
difficult it is to compute the Voronoi diagram,
or an approximation of it, from a set of queries.

\ifFull
\subsection{Related Work}
Let us briefly review some related prior work.
\fi

\ifFull
Bancilhon and Spyratos~\cite{bs-pirdd-77},
Deutsch and Papakonstantinou~\cite{dp-pidp-05},
and Miklau and Suciu~\cite{ms-faidd-07}
provide general models for
privacy loss in information releases from a database, 
called \emph{query-view security},
which identifies the sensitive information as a
specific secret, $S$.
Attackers are allowed to form legal queries and ask them of the database,
while the database owner tries to protect the information that these
queries leak about the secret $S$.
Note that we are instead considering
a related, but different, quantitative measure, 
where there is no specifically sensitive part of the
data, but the data owner, Alice, is trying to limit releasing 
too much of her data.

There has been considerable recent work on designing 
technological approaches that can help protect the privacy or
intellectual property rights of a database by modifying its content.
For example, several researchers
(e.g., see~\cite{ak-wrd-02,ahk-swrd-03,ahk-wrdfa-03,ak-wrd-02,g-qpwrd-03,%
sv-hcwsd-04,sap-rprd-03,s-rard-07})
have focused on data watermarking, which is a
technique for altering the data to make it easier,
after the fact, to track when someone has stolen information
and published their own database from it.
Alternatively, several
other researchers~\cite{lefevre05,samarati01,samarati98,mw-cok-04,%
afkmptz-at-05,bkbl-ekuct-07,zyw-pekcd-05,ba-dpoka-05} have proposed
\emph{generalization} as a way of
specifying a quantifiable secrecy-preservation requirement for databases.
The generalization approach is to group attribute values into equivalence
classes, and replace each individual attribute value with its class name,
thereby limiting the information that can be derived from any selection query.
Our assumption in this paper, however, is that the data owner, Alice, is not
interested in modifying her data, since she derives an economic interest from
its accuracy.
She may be interested instead in placing a reasonable limit
on the number of queries any one user might ask, so that she can
limit her exposure to the risk of that user cloning her data.

There has been less prior related work that takes the ``black hat'' perspective 
of this paper, which asks the question of how quickly a data set can be
discovered from seemingly minimalistic responses to queries on that object.
For example,
Goodrich~\cite{g-magd-09} studies the 
problem of discovering a DNA string from genomic comparison queries and
Nohl and Evans~\cite{ne-qiltb-06} 
study quantitative measures on the global system
information that is leaked from learning
the contents of multiple RFID tags that are generated by that system.
\fi

Motivated by the problem of having a robot discover the shape of an
object by touching it~\cite{cy-sp-87}, 
there is considerable amount of related work in
the computational geometry literature on discovering polygonal and
polyhedral shapes from probing 
\ifFull
(e.g., see~\cite{aby-nccr-90,aiir-pshlr-93,bs-ppmih-93,by-psncp-92,%
dey-pcp-86,es-pcpxr-88,js-mbpsc-92,rs-tos-95,s-irgp-92,s-pcphp-91}).
Rather than review all of this work in detail, we
\fi
refer the interested reader to the survey and book 
chapter by Skiena~\cite{s-pgp-89,s-grp-97}, and simply mention
that, with the notable exception of work by 
Dobkin {\it et al.}~\cite{dey-pcp-86},
this prior work
is primarily directed at discovering obstacles in a two-dimensional environment
using various kinds of contact probes.
\ifFull
Thus, although it is related, most of this 
prior work cannot be adapted to the problem of
discovering a Voronoi diagram through nearest-neighbor probes.

By a well-known lifting method
(e.g., see \cite{bkos-cgaa-97}), a 2-dimensional Voronoi diagram can
be defined by a projection to the plane of the skeleton of a 
3-dimensional convex polyhedron, which is defined by the upper
envelope of planes tangent at point sites mapped vertically to a 
certain paraboloid.
This property of Voronoi diagrams implies 
that the method of Dobkin {\it et al.}~\cite{dey-pcp-86} for 
discovering a convex polytope via ray-shooting ``finger''
probes can be used to discover a Voronoi diagram in
the plane using nearest-neighbor queries (for instance, the queries
we call ``exact queries'').
\fi
Translated into this context,
their method results in a scheme that would use $7n-5$ queries to
clone a Voronoi diagram, with a time overhead that is $\Theta(n^2)$.

In the framework of 
\ifFull
\emph{retroactive data structures}~\cite{b-sedop-08,dil-rds-07,gk-odvrs-09},
which are themselves related to
\emph{persistent data structures}~\cite{dr-par-91,dsst-mdsp-89,st-pplup-86},
\fi
each update operation $o$ to a data structure $D$, such as an insertion or deletion of an element, comes with a unique numerical value, $t_o$, specifying a
\emph{time value} at which the operation $o$ is assumed to take place.
The order in which operations are presented to the data structure is not assumed to be the same as the order of these time values. 
Just like update operations, query operations also come with time values; a query with time value $t$ should return a correct response with
respect to a data structure on which all operations with $t_o<t$ have been performed.
Thus, an update operation, $o$, having a time value, $t_o$, will affect
any subsequent queries having time values greater than or equal to $t_o$.
In a \emph{partially retroactive} data structure the time for a query must be at least as large as the maximum $t_o$ seen so far, whereas
in a \emph{fully retroactive} 
data structures there is no restriction on the time
values for queries.
\ifFull
Persistent data structures~\cite{dr-par-91,dsst-mdsp-89,st-pplup-86}
also allow for such 
``queries in the past,'' as with fully retroactive data structures, but 
in a persistent structure, $D$, all
updates must either be for the current version of $D$ or 
they must ``fork'' off an alternate version of $D$ from a 
past instance of $D$.

\fi
Demaine {\it et al.}~\cite{dil-rds-07} show
how a general comparison-based ordered dictionary 
(with successor and predecessor queries) of $n$ elements (which may not
belong to a total order, but which can always be compared when they are in
$D$ for the same time value) 
can be made fully retroactive in $O(n\log n)$ space
and $O(\log^2 n)$ query time and amortized $O(\log^2 n)$
update time in the pointer machine model.
Blelloch~\cite{b-sedop-08} and
Giora and Kaplan~\cite{gk-odvrs-09}
improve these bounds, for numerical (totally ordered) items,
showing how to achieve a fully retroactive ordered dictionary 
in $O(n)$ space and $O(\log n)$ query and update times in
the RAM model.
These latter results do not apply to the general comparison-based
partially-ordered setting, however.

\ifFull
\subsection{Our Results}
\fi
Given a set $S$ of $n$ points in the plane, with an API
that supports nearest-neighbor queries, 
we show how queries of 
Type~1 and Type~2 allow an exact cloning
of the Voronoi diagram, $V(S)$,
of $S$ with $O(n)$ queries and $O(n \log^2 n)$ processing time. 
Our algorithms are based on non-trivial modifications 
of the sweep-line algorithm
of Fortune~\cite{f-savd-87} (see also~\cite{bkos-cgaa-97})
so that it can construct a Voronoi diagram
correctly in $O(n\log^2 n)$ time while
tolerating unbounded amounts of backtracking.
We efficiently accommodate this 
unpredictable backtracking through the use of a fully 
retroactive data structure for general comparison-based dictionaries.
In particular, our method is based on our showing that the 
dynamic point location method of Cheng and Janardan~\cite{cj-nrdpp-92} 
can be adapted
into a method for achieving a general comparison-based fully retroactive ordered
dictionary with $O(n)$ space, $O(\log n)$ 
amortized update times, and $O(\log^2 n)$
query times.
We also provide lower bounds that show that, even with an adaptation of the 
Dobkin {\it et al.}~\cite{dey-pcp-86} approach optimized for nearest-neighbor
searches, there is a sequence of query responses that requires $\Omega(n^2)$
overhead for their approach applied to these types of \emph{exact} queries.
Nevertheless, we show that it is possible to clone
$V(S)$ using only $3n$ queries.
We prove that queries of Type~3 can never exactly 
clone $V(S)$, however, nor even determine 
with certainty the value of $n=|S|$.
Nevertheless, we show that with $n \log(\frac{1}{\epsilon})$
queries we can construct an $\epsilon$-approximate cloning of $V(S)$ 
that will support 
approximate nearest neighbor queries guaranteeing a response 
that is a site within (additive) $\epsilon>0$ distance
of the exact nearest neighbor of the query point.  
\ifFull
Finally, we show that an approach based on using a quadtree to approximate
the Voronoi diagram, which is a natural alternative approach,
in this case results in a suboptimal approximation.
\fi

\ifFull
\section{A Fully-Retroactive General Comparison-Based Ordered Dictionary}
\fi
In this section, we develop a fully retroactive
ordered dictionary data structure using $O(n)$ space, $O(\log n)$ amortized
update time, and $O(\log^2 n)$ query time, based on a dynamic point location method of
Cheng and Janardan~\cite{cj-nrdpp-92}. 
\ifFull
Our data structure only performs comparisons between pairs of items that 
are simultaneously active at some point in the history of 
the retroactive data structure.
\fi
The main idea is to construct an interval tree, $B$,
over the intervals between the insertion and deletion times of each item in the dictionary, and to
maintain $B$ as a BB[$\alpha$]-tree~\cite{m-mdscg-84}.

Each item $x$ is stored at the unique node $v$ in $B$ such that $x$'s insertion
time is associated with $v$'s left subtree and $x$'s deletion time is
associated with $v$'s right subtree.
We store $x$ in two priority search trees~\cite{m-pst-85}, 
$L(v)$ and $R(v)$,  associated with node $v$.
These two priority search trees are both ordered by the dictionary ordering of the items stored in them; all such items are active at  the time value that
separates $v$'s left and right subtrees, so they are all comparable to each other. The priority search trees differ, however, in how they prioritize their items.
As with priority search trees more generally, each node in $L(v)$ and $R(v)$ stores two items, one that is used as a search key and another that has the minimum or maximum priority within its subtree. In $L(v)$, the insertion time of an item is used as a priority, and a node in $L(v)$ stores the item that has the minimum insertion time among all items within the subtree of descendants of that node. In $R(v)$, the deletion time of an item is used as a priority, and a node in $R(v)$ stores the item that has the maximum deletion time within its subtree.

An insertion of an item $x$ in $D$ is done by finding the appropriate node $v$ of the interval tree and inserting $x$ into $L(v)$ and $R(v)$, and a deletion is likewise done by
deletions in $L(v)$ and $R(v)$.
Updates that cause a major imbalance in the interval tree structure are
processed by rebalancing, which implies, by the properties of
BB[$\alpha$]-trees~\cite{m-mdscg-84}, that updates run in $O(\log n)$
amortized time.

Queries are done by searching the interval tree for the nodes with the property that the retroactive time specified as part of the query could be contained within one of the time intervals associated with that node. For each matching interval tree node~$v$, we perform a search in either $L(v)$ or $R(v)$ depending on the relation between the query time and the time that separates the left and right children of $v$. The search method of Cheng and Janardan~\cite{cj-nrdpp-92} allows us to find the successor of the query value, among the nodes stored in $L(v)$ or $R(v)$ with time intervals that contain the query time, in time $O(\log n)$. The result of the overall query is then formulated by comparing the results found at each interval tree node and choosing the one that is closest to the query value. Thus, the query takes $O(\log n)$ time to identify the interval tree nodes associated with the query time, $O(\log^2 n)$ to query each of logarithmically many priority search trees, and $O(\log n)$ time to combine the results, for a total of $O(\log^2 n)$ time.

\begin{theorem}
One can maintain a fully-retroactive general comparison-based dictionary
on $n$ elements, using $O(n)$ space,
so that updates run in $O(\log n)$ amortized time and 
predecessor and successor 
queries run in $O(\log^2 n)$ time.
\end{theorem}

\section{Exact Query Probes}

We begin our study of Voronoi diagram cloning 
with the strongest sort of queries---Type~1. 
Given a query point $p$, a Type-1 query returns
the site $q$ in $S$ nearest to $p$, that is,
it returns the geometric location of $q$, $p$'s nearest-neighbor in $S$. 
In the event that
$p$ has more than one nearest neighbors in $S$, 
all nearest neighbors are returned.
We show that only $O(n)$ queries and $O(n \log^2 n)$ processing
time is needed to completely clone $V(S)$---which, 
as implied, also means we explicitly
have determined both $S$ and $n$.

\ifFull
In this section, we present an algorithm
that accomplishes this task and we argue that the algorithm
is correct and requires only $O(n \log^2 n)$ time. 
We show that the number of probes needed is less than $4n$.  
\fi

\paragraph{\bf Overview of Our Algorithm.}
Our algorithm is adapted from the plane sweep Voronoi diagram algorithm of
Fortune~\cite{f-savd-87}, with a significant modification 
to allow for unbounded and
unpredictable amounts of backtracking.
The fundamental difference 
is that the Fortune algorithm begins with the set of sites, $S$,
completely known; in our case, the only thing
we know at the start is a bounding box containing $S$. 
Using the formulation of
de~Berg {\it et al.}~\cite{bkos-cgaa-97},
Fortune's algorithm uses an event queue to 
controls a sweep line that moves in order
of decreasing $y$ coordinates, 
with a so-called ``beach line''---an 
$x$-monotone curve made up of parabolic segments
following above the sweep line. 
The plane above the beach line is partitioned into cells 
according to the final Voronoi 
diagram of $S$. 
There are two types of events, caused 
when the sweep line crosses point sites in $S$
and Voronoi vertices in $V(S)$; the latter points
are determined as the algorithm progresses. 
In our version, we need to find both 
the sites and the Voronoi vertices as the plane sweep advances. 
And because not all sites are known in advance,
we will need to verify \emph{tentative} Voronoi vertex events 
as we sweep across them, at times
backtracking our sweep line when our queries reveal new sites 
that invalidate tentative 
Voronoi vertices and introduce new events that are actually 
above our sweep line.
We will show that each query discovers a feature in the Voronoi diagram,
that the number of times we backtrack 
is bounded by the number of these features, and 
these facts imply that
the number of queries and updates we perform in a retroactive 
dictionary used to implement our sweep-line algorithm is $O(n)$.
In fact, we will prove 
that the number of probes is at most $4n$.

We begin with an overview of our algorithm.
The algorithm begins by finding all the Voronoi regions 
and edges that intersect the top edge 
of the bounding box, $B$. 
If there are $k$ such regions (and thus $k-1$ edges), this can 
be accomplished
in $O(n)$ time with $2k-1$ queries. 
This step initializes our event queue with $k$ of the point sites in $S$.

The algorithm then proceeds much as the Fortune algorithm, 
but with the following two important changes.  
Whenever we reach a point site 
event
for some site $q \in S$ (i.e., when $q$ is removed from the event queue), we do a nearest-neighbor query on the point of the 
beach line directly above $q$---that is, the point with the 
same $x$-coordinate as $q$ and a $y$ coordinate
on the beach line that exists for the time value when the sweep line hits $q$. 
The position of this query point can be determined by using a retroactive dictionary
queried with respect to
the components of the beach line for the time value (in the plane sweep)
associated with the point $q$.
(See Fig.~\ref{fig:sweep}.)
Querying this point will either confirm a 
Voronoi edge known to be part of the final Voronoi diagram 
(in which case we proceed with the 
sweep) or it will discover a new site $r$ in $S$ 
(in which case the sweep line restores point
$q$ to the event queue and backtracks to $r$).  

\begin{figure}[b!]
\vspace*{-16pt}
\begin{center}
\includegraphics[width=3.5in]{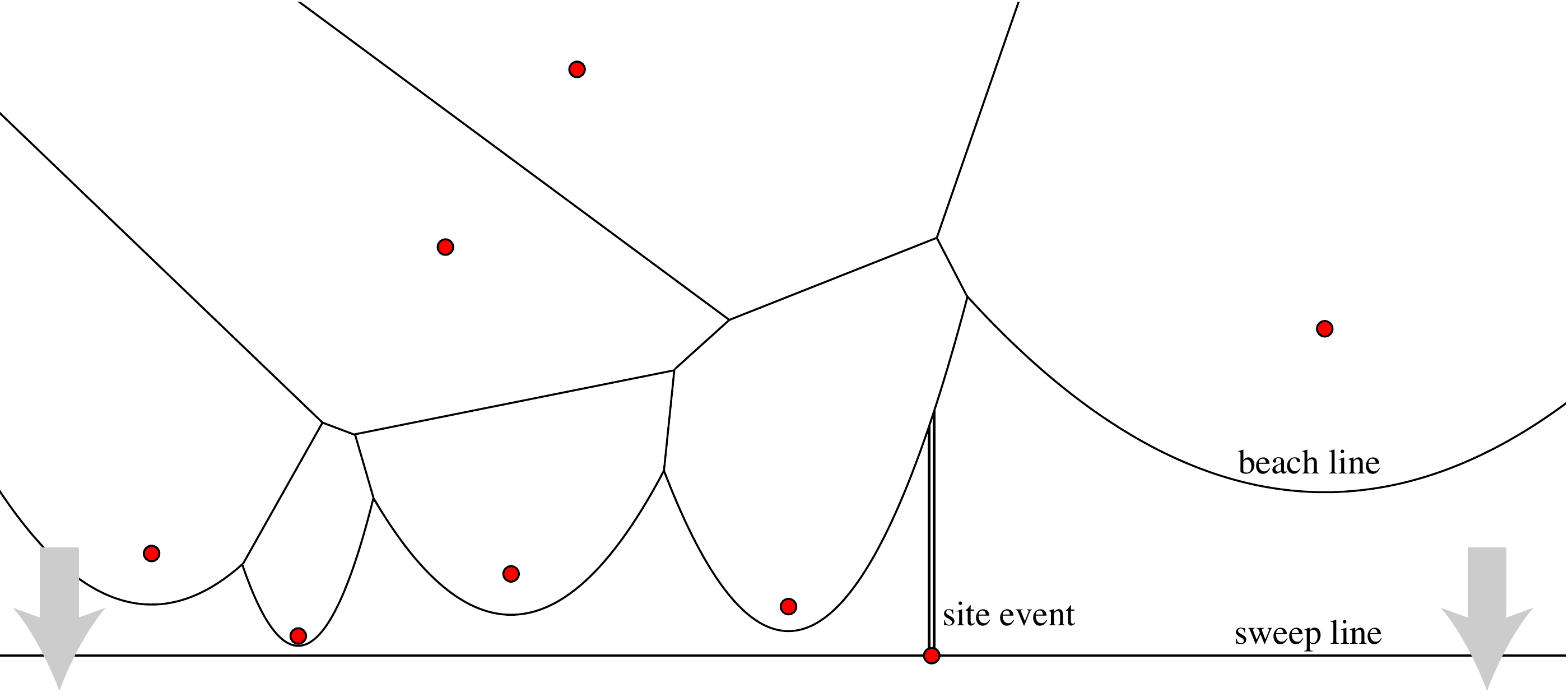}
\end{center}
\vspace*{-12pt}
\caption{Illustrating the sweep-line algorithm for constructing 
a Voronoi diagram.}
\label{fig:sweep}
\end{figure}

The second type of event is a 
tentative Voronoi vertex event.  
We do a nearest-neighbor query at the point believed to be a
Voronoi vertex, which either confirms
the vertex and all of its adjacent Voronoi edges above it, 
or it discovers a new site in $S$.
This new site must be above the sweep line: at the time that Fortune's algorithm processes a Voronoi vertex event its sweep line must be as far from the Voronoi vertex as the three sites generating the vertex, so undiscovered sites below the sweep line cannot be nearest neighbors to the tentative vertex.
If the algorithm discovers a new site above the sweep line, 
then again we backtrack and 
process that new site.
In either case the Voronoi vertex is 
removed from the event queue---either added to the
Voronoi diagram being constructed as a validated vertex, 
or ignored as a false vertex.
\ifFull
We give the details in an appendix.
\fi

\paragraph{\bf Correctness and Complexity.}
Both the correctness and the analysis of this algorithm
make use of the following important
observation. Though the algorithm backtracks 
at certain ``false'' events---or tentative
events that are proven false---it never 
completely removes any Voronoi components that have been confirmed by probes.
 Voronoi edges can only be added in
two ways: the addition of a new site that creates one new edge, 
or the addition of a Voronoi 
vertex that terminates two edges and creates one new edge. 
In both cases, the edge
is verified as an actual edge using a query before it is 
added to the Voronoi diagram
being constructed, thus the diagram never contains edges that could 
later be falsified.
(See Fig.~\ref{fig:backtrack}.)

\begin{figure}[b!]
\begin{center}
\includegraphics[width=3.5in]{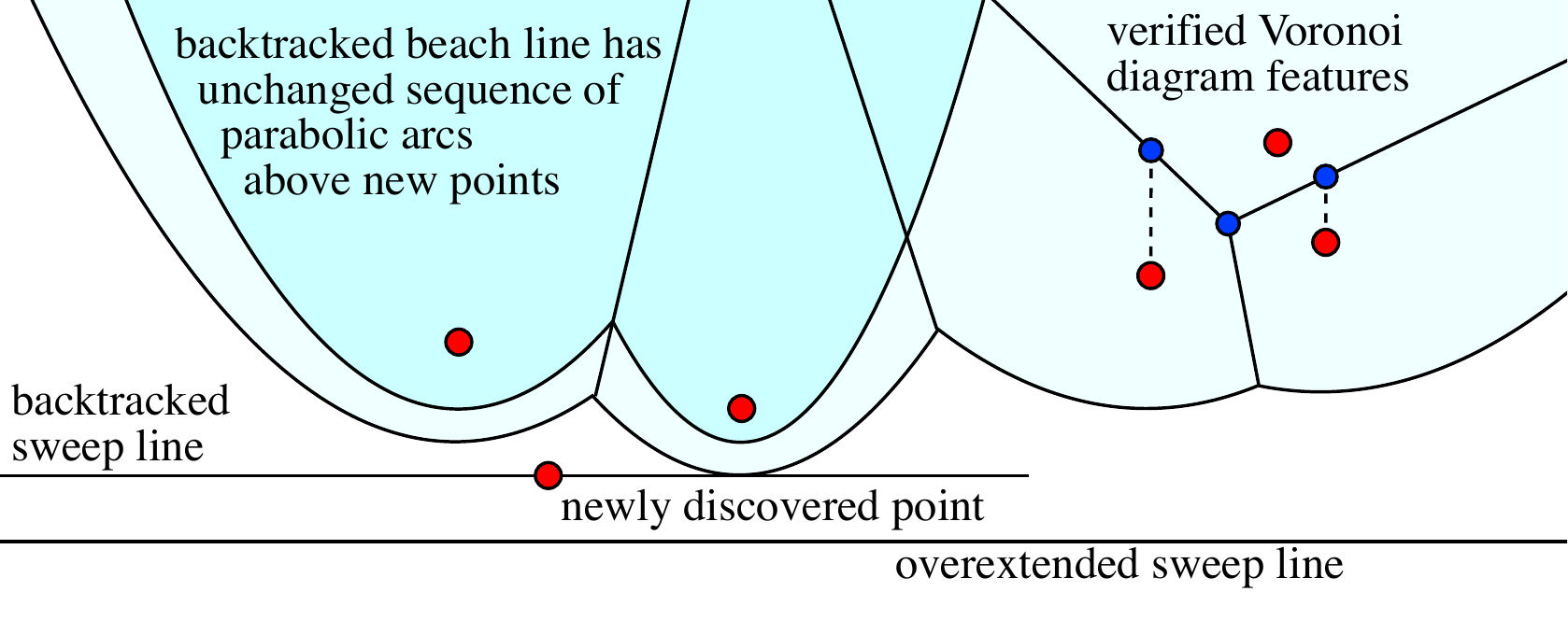}
\end{center}
\vspace*{-12pt}
\caption{Backtracking the sweep-line.}
\label{fig:backtrack}
\end{figure}

The insertion of a new site begins a new edge directly above it, 
where the parabola
of the site being added to the tree $T$---a degenerate line-segment 
parabola at the instant it
is added---intersects the existing parabola above it, 
thus replacing one leaf in the tree with
three.  
But before this site is inserted with its edge, the edge is 
tested with a query into
the existing Voronoi diagram.  
The other time an edge may be added is at
a Voronoi vertex where two existing edges meet and 
a third new one is created. But all
tentative Voronoi vertices are also verified by queries 
before they become circle events.

\ifFull
Because of these verifying queries, the correctness of the algorithm follows from the correctness 
of the original Fortune algorithm and the fact that every tentative Voronoi vertex is verified by a 
query before being added to the diagram. Verifying every vertex in turn implies that every site in 
$S$ is discovered. In particular, suppose we constructed a Voronoi diagram $V(S)$ for some set 
$S'$ that was missing at least one site from $S$. (Since queries never return non-existing sites, 
$S'$ cannot have any additional sites not in $S$.) Let $p \in S$ be any such site missing from 
$S'$. At least one Voronoi vertex in $V(S')$ must be closer to $p$ than to any site in $S'$. 
This follows, for example, from the fact that every Delaunay triangle has an empty 
circumcircle corresponding to a Voronoi vertex, so at least the Delaunay triangle containing 
$p$ would have a corresponding Voronoi vertex closer to $p$.  And so a query on that vertex 
would return $p$ as the closest site.
\fi

There are a few key observations that will lead to the analysis of the algorithm's run time
and total number of queries. First, the sweep line will only backtrack when a new site in 
$S$ is discovered, and so there are at most $n$ backtracks. Second, every time we have a tentative 
Voronoi vertex that turns out to be unverified---that is, an event that turns out not to be part of the 
final diagram---we have also discovered a new site in $S$, and thus we have at most $O(n)$ 
phony events that are processed. It follows that the run time of the algorithm is asymptotically 
equivalent to the original Fortune algorithm, modulo the time needed for our retroactive data structure queries.
Furthermore, after our initialization stage, every nearest-neighbor query either finds a site, verifies a site by
looking at the Voronoi edge above it, or verifies a Voronoi vertex. The initialization, as noted,
requires $2k-1$ queries if there are $k$ sites initially discovered. Queries discover $n-k$ more sites,
verify $n$ sites, and find at most $2n-2-k$ Voronoi vertices, for a total of $4n-3$ which is less
than $4n$ queries. 

The algorithm requires the same run time as the original Fortune plane sweep except
for the processing of the tentative Voronoi vertices that prove to be phony,
and all the backtracking (which is implemented using our retroactive
dictionary). As noted, there are
$O(n)$ of these backtracking steps, 
since these can only occur once for each previously undetected site in $S$.
So the overall number of updates and queries in our
sweep-line-with-backtracking algorithm is $O(n)$; hence,
the running time of our algorithm is $O(n \log^2 n)$.
Thus, we have the following.

\begin{theorem}
Given a set $S$ of $n$ points in $\R^2$, we can construct a copy of
$V(S)$ using at most $4n$ 
Type-1 queries and $O(n\log^2 n)$ time.
\end{theorem}

\ifFull
\subsection{An Alternate Algorithm Using More Time and Fewer Queries}
It is possible to further limit the 
number of queries at the cost of increased worst-case computational
complexity.  
If $n>1$, we can also exactly clone $V(S)$ 
with at most $3n$ queries if we are willing to
spend worst case $O(n^{2})$ time implementing this alternative.

\fi
This second algorithm follows an incremental construction paradigm,  
based on the general approach 
of Dobkin {\it et al.}~\cite{dey-pcp-86} for discovering a
3-dimensional convex polyhedron using finger probes.
The alternative algorithm begins by querying each of 
the four corners of the bounding box.  
There are three cases to consider:
these probes may discover one, two, or more than two
sites in $S$.  
If we discover more than two sites, then we construct the Voronoi diagram
of all $3$ or $4$ of the sites discovered by these four queries, 
but we mark each Voronoi
vertex as tentative and put it into a queue. 
The algorithm then proceeds as follows until the queue
is empty. Remove a tentative Voronoi vertex 
from the queue, and query it. If the query
reveals that it is a Voronoi vertex---that is, 
it has the three expected nearest neighbors---then we
confirm the vertex and continue. 
If it is not a Voronoi vertex, then it must be closer to some
previously undiscovered site in $S$, 
that will be returned by the query. We add that site
to our list of known sites and update the 
Voronoi diagram in worst case $O(n)$ time
using incremental insertion. When the queue is 
empty, we have a complete Voronoi diagram.
Every probe except possibly one of the 
four corner probes discovered either a new site
in $S$ or confirmed a Voronoi vertex, 
and so the total number of queries is at most $n+(2n-5)+1 <  3n$.
If the four corner queries discovered only two sites, 
then we compute the Voronoi edge
that would be shared by these two sites 
if they were the only two sites in $S$, and we
query both intersections of this edge with the bounding box.  
If we confirm both edges,
then $n=2$ and we are done. 
We have used $3n=6$ queries. If at least one of these
two additional queries discovers another site, 
then we have at least three known sites 
and we proceed as with the previous case. 
Every query except at most three of the initial
queries either confirmed a Voronoi vertex 
or discovered a new point site, and so the
total number of queries is at most $n+(2n-5)+3 < 3n$.
All four corners will belong to the same 
Voronoi region if and only if there is only one site
in the bounding box, in which case $4$ queries was sufficient.
\ifFull
Thus,
we have the following.
\fi

\begin{theorem}
\label{thm:type-one}
Given a set $S$ of $n$ points in $\R^2$, we can construct a copy of
$V(S)$ using at most $3n+1$ Type-1 queries and $\Theta(n^2)$ time.
\end{theorem}
\begin{proof}
We have already established the quadratic upper bound.
For the sake of a lower bound, imagine that we have a set $S'$ of 
$n/2$ points on the bottom boundary of $B$, all within distance $\delta$
of the point $(0,0)$, for a small parameter $\delta$ with
$0<\delta\le 1/2^n$. These
points, by themselves,
construct a Voronoi diagram with parallel edges. Suppose further
that there is a single point, $p_0=(\delta,1-\delta)$, 
near the top boundary of $B$. 
The Voronoi region for $p_0$ intersects the Voronoi region of every point in
$S'$.
The above algorithm therefore, after discovering the boundary
points in $S'$, would next
query a vertex on the Voronoi diagram $V(S')$ of $\cal B$, which will
discover $p_0$. 
Next it will probe at vertex that is equi-distant to $p_0$ and a
point in $S'$.
Suppose that this probe discovers a point $p_1=(\delta,1/2)$.
This point is closer to every point in $S'$ than $p_0$; hence,
updating the Voronoi diagram to go from
$V(S'\cup \{p_0\})$ to
$V(S'\cup \{p_0,p_1\})$ takes $\Omega(n)$ time.
Now, suppose querying a tentative vertex of
$V(S'\cup \{p_0,p_1\})$, which will be equi-distant from $p_1$ and a
point in $S'$, discovers a point $p_2=(\delta,1/2^2)$.
Again, updating the Voronoi diagram takes $\Omega(n)$ time.
Suppose, therefore, that
we continue in this way, with each newly-discovered point
$p_i=(\delta,1/2^i)$ requiring that we spend $\Omega(n)$ time to
update the current Voronoi diagram.
After discovering $p_{n/2-1}$, the $n$-th point in the set 
$S=S'\cup \{p_0,p_1,\ldots,p_{n/2-1}\}$,
we will have spent $\Omega(n^2)$ time in total to
discover the Voronoi diagram $V(S)$.
\qed
\end{proof}

\section{Distance Query Probes}
We next consider our Voronoi diagram cloning algorithm for the case when we use
use only distance query probes---probes
that return the distance to and label of the nearest site(s).
\ifFull
Our method follows a similar approach to the backtracking method we used above.
The main difference is that we require extra work to determine the exact location of the 
sites in $S$.

\fi
We begin by describing how we can find those sites whose Voronoi regions (and thus also edges)
intersect the top boundary of the bounding box where the sweep-line begins. We will speak of
a \emph{probe circle} as the set of possible locations of a site returned by a probe $p$: it is
the circle of radius $d$ centered at $p$, where $d$ is the distance returned to the nearest site.

\paragraph{\bf Initializing the Sweep Line.}
We begin the initialization process by probing at the two top corners of the bounding box.
If both probes return the same site $p$, then by convexity of Voronoi regions the entire
top edge of the box belongs to the Voronoi region $V(p)$. Furthermore, both probes
also return a distance $d$  to the site $p$, and so $p$ must fall on an intersection of two circles
of radius $d$ centered at the two corner probe locations. Since one of these intersections
is above the bounding box, the remaining intersection gives the exact location of $p$.

Assume that the two corner probes $p_{l}$ and $p_{r}$ on the left and right respectively return 
different sites $q_{l}$ and $q_{R}$ respectively.  We know the distance $p_{l}$ to $q_{l}$ and 
$p_{R}$ to $q_{R}$, but we don't know the exact locations of the two sites, nor do we know if there 
are any other sites with regions intersecting the bounding box. The segment between these two 
probes is  therefore not fully classified. We will describe a recursive procedure for classification.

Let $L$ be an unclassified segment, with the probes $p_{l}$ and $p_{r}$ on the left and right 
sides of the segments respectively returning different sites $q_{l}$ and $q_{r}$.
First, we probe the midpoint $p_{m}$ of $L$. The probe returns either one of the two known sites,
or a new site $q_{m}$.  If it returns a new site, then we divide $L$ into two segments that
are both unclassified, but which have classified endpoints $q_{l}$, $q_{m}$, and $q_{r}$, 
and we recursively classify them.

Suppose query $p_{m}$ returns one of the already known sites. If this probe returns 
$q_{r}$, then we can
immediately compute the exact location of $q_{r}$ from the two probes $p_{r}$ and $p_{m}$, 
since only one of the intersections of the probe circles is inside the bounding box. We also 
have classified the segment $L_{r}$ between probes $p_{r}$ and $p_{m}$ as being fully inside 
the Voronoi region of $p_{r}$. 
The other half of the segment, $L_{l}$, however, is not classified; we only know the Voronoi
regions of its endpoints are the regions of $q_{r}$ and $q_{l}$. Here it would be tempting to 
again probe the midpoint of the segment $L_{l}$;  however that could lead to an unbounded 
number of probes as we repeatedly divide the segment in half because the midpoint of the
 remaining unclassified segment $L_{l}$ could still belong to $q_{r}$ and so we would gain
 no new information about $q_{l}$. What we do instead
 is use the known location of $q_{r}$, which must be outside the probe circle at
 $p_{l}$, to find a probe location $p_{l2}$ close enough to $p_{l}$ that
 it is guaranteed not to gives us $q_{r}$. This is possible since the perpendicular bisector between 
 $q_{r}$ (which is a known point) and any point on the probe circle from $p_{l}$---which
 is the set of candidate locations for $q_{l}$---must fall between $p_{m}$ and $p_{l}$ on
  a finite segment computable in $O(1)$ time, and thus
 anything between that range and $p_{l}$ is closer to $q_{l}$ than to $q_{r}$. 
 
So this new probe $p_{l2}$ returns either $q_{l}$ or a new site.  If it returns a new site, then we divide 
the unclassified segment into two unclassified segments and recursively classify them. 
If this new probe gives us $q_{l}$ then we now know the exact location of $q_{l}$ from two probes. 
From the exact locations of $q_{r}$ and $q_{l}$, we can compute and probe
 where their Voronoi edge ought to cross the bounding box, either confirming that 
 Voronoi edge---which means that the entire edge is now classified---or we discover a new
 site. If the probe gives us a new site, then again we divide the unclassified segment into 
 two unclassified segments and recursively classify them.
 
 \begin{lemma} The initialization stage for the sweep line requires $O(k)$ time and $3k-1$ probes 
 where $k$ is the number of sites whose Voronoi regions intersect the top of the bounding box.
 \label{lem:init-dist}
  \end{lemma}
 
\begin{proof}
 Note that every probe either identifies a previously undiscovered site ($k$ probes), provides
a second probe with more information on an already discovered site enabling the exact location 
of this site to be computed ($k$ probes), or confirms a Voronoi edge ($k-1$ probes for $k$ 
regions). So the total number of probes in this section is $3k-1$ where $k$ is the number of sites 
whose Voronoi regions intersect the top of the bounding box.  Each probe is processed in
$O(1)$ time.
\qed
\end{proof}
 
\paragraph{\bf Processing the Sweep Line.}
The previous subsection explains how to initialize the sweep line. The algorithm making use
of distance-only queries now proceeds as with the exact query probe version of the previous
section, except that a slightly different approach requiring more probes will be needed to
process tentative site events.

As with the algorithm of the previous section, there are two types of tentative events: a tentative
Voronoi vertex for three known sites, and a tentative Voronoi edge that falls directly above a known 
site and is determined by one other known site.
Both of these events need to be verified by probe--that is, we need to determine if these events
are actually real, or whether there is some other site closer to the events. 
In both cases, we use a probe $p$ where the tentative
Voronoi feature \emph{should} be. If that problem returns the correct three or two site labels 
(at the correct distance), then the verification is complete, and we proceed as with the
algorithm of the previous section.

However, these probes may discover a new site $q$; in this case, they give only the distance to that site and not its actually location. We
need two more probes that return the same site in order to discover its exact location---but 
these probes may instead return yet other new sites. We now describe how to choose the 
locations of these probes so that no work is wasted, and each probe either
verifies a Voronoi vertex, verifies a Voronoi edge above a known site, or is one of  three
probes that exactly locates a site.

Let $p_{1}$ be a probe during the sweep line, that attempts to verify a Voronoi vertex or edge,
and instead discovers a new site $q_{1}$ that was not previously known.  Let $d=d(p_{1},q_{1})$ be 
the distance returned from probe $p_{1}$ to its site $q_{1}$, and let $e$ be the distance from
$p_{1}$ to the nearest previously known sites---that is, the two or three sites whose tentative
Voronoi vertex or edge it was seeking to verify.  Since probe $p_{1}$ returned $q_{1}$, we know
that $d < e$. Let $p_{2}$ be any probe location such that $d(p_{1},p_{2}) <  \frac{e - d}{2}$.
By the triangle inequality,  we know $d(p_{2},q_{1}) < d +  \frac{e - d}{2} = \frac{e + d}{2}$,
while $d(p_{2},r)>e-\frac{e - d}{2}= \frac{e + d}{2}$, where $r$ is any of the two or three previously 
known closest sites to $p_{1}$. It follows immediately that probe $p_{2}$ cannot return any 
previously known site except $q_{1}$ which was first discovered by probe $p_{1}$. We can 
choose any probe location meeting this restriction, $d(p_{1},p_{2}) <  \frac{e - d}{2}$, which
is computable in $O(1)$ time.

So there are two possibilities with probe $p_{2}$: either it returns site $q_{1}$ again, or
it returns a new site $q_{2}$.  If $p_{2}$ returns $q_{1}$, we now have two probes returning
that site, and distances to that site, so its location is one of at most two intersections between the
two probe circles.  We can now probe either one of those two intersections, and from the
result we determine the exact location of $q_{1}$ because the probe either returns $q_{1}$
at distance $0$, or it returns some other site, or it returns $q_{1}$ at a distance $>0$.

If $p_{2}$ returns a new site $q_{2}$, then we now have two sites that have been discovered,
but whose exact locations are not known.  We can discover the exact location using the 
recursive method of the previous subsection, treating the segment $p_{1}p_{2}$ as an unclassified 
segment, probing its midpoint, and continuing. However, once we have 
received a site as the result of two probes, we still require a third probe to exactly locate it
since both intersections of the first two probe circles might be inside the bounding box.

 \begin{lemma} Processing the remaining events (after the initialization) for the sweep line 
 requires at most $6n-3k-5$ probes 
 where $k$ is the number of sites whose Voronoi regions intersect the top of the bounding box.
 \label{lem:sweep-dist}
 \end{lemma}
 
\begin{proof}
There are $n-k$ sites to be discovered, $n$ sites that need to have an edge verified above
them, and at most $2n-5$ Voronoi vertices in the Voronoi diagram of $n$ sites.
Every probe accomplishes one of five things: it verifies a Voronoi vertex ($2n-5$ probes),
verifies a Voronoi edge directly above a site ($n$ probes), or is one of exactly three 
probes used to discover and then exactly locate a new site ($3(n-k)$ probes.)  The
total number of probes required is therefore at most $6n-3k-5$.
\qed
\end{proof}

\ifFull
\paragraph{\bf Correctness and Analysis.}
The correctness of the algorithm follows from the same argument as in the previous section,
as long as we know that three probes suffice to determine the exactly location of a site.
The analysis of the run-time is straightforward.
\fi

 \begin{theorem} 
Given a set $S$ of $n$ points in $\R^2$, we can construct a copy of
$V(S)$ using at most $6n-6$ Type-2 queries and $O(n\log^2 n)$ time.
  \end{theorem}
 
 \ifFull
\begin{proof}
The bound on the number of probes follows directly from Lemmas~\ref{lem:init-dist} 
and~\ref{lem:sweep-dist}.  Since we can compute the exact location of each
site with three probes instead of one, the run time of these algorithm
 is asymptotically equivalent to that of the previous section.
\qed
\end{proof}
\fi

\section{Label-Only Query Probes}

\ifFull
In this section we explore what information can be gained 
from the weakest of our queries (Type-3): 
given a query point $p$, the query returns only the label 
of the site $q$ in $S$ that is nearest to 
$p$, but no location or proximity information about site $q$.
\fi
Using queries 
of the third type, it is impossible to 
\emph{exactly} clone $V(S)$ or even to determine 
with certainty the value of $n=|S|$.   However even with this 
minimal query information, we construct an approximate Voronoi diagram $V(S')$ which, without \emph{explicitly} storing 
the locations of the sites in $S$, will still support later arbitrary approximate proximity queries to 
$V(S)$. 
\ifFull
That is, for a new query point $p'$---not necessarily one of the query points used to construct the 
data structure---$V(S')$ will be able to determine a site from $S$ that is within $\epsilon$ of the distance to the nearest neighbor of $p'$, where $\epsilon$ can be made arbitrarily small.

We describe first an approach that uses a common 
geometric search structure, the quadtree, and then show how we can improve upon these results
and compute a Voronoi diagram.  Neither approach is guaranteed to find all the sites in $S$,
but both will 
distinguish between any two sites separated by at least $\epsilon$.
Let $N$ be the number of sites found. The quadtree approach uses $O(N \frac{1}{\epsilon})$ queries
and time, while our approach of cloning an approximate Voronoi diagram requires 
 $O(N \log(\frac{1}{\epsilon}))$ queries and 
 $O(N (\log N + \log(\frac{1}{\epsilon})))$ time. 
 So the Voronoi diagram cloning approaches 
 reduce the multiplicative factor from $\frac{1}{\epsilon}$
 to $\log(\frac{1}{\epsilon})$.

 First, however, we show the impossibility of constructing an exact clone.
 
\subsection{Privacy Under Label-Only Queries}
The main result of this subsection is the following lemma, 
the proof of which is straightforward
but is included here for completeness.

\begin{lemma}
Given a set $S$ of $n$ planar points
in a bounding box of area greater than $0$, 
using any finite number of queries $p$ that return only the unique label $q$ of the closest site 
(or sites) in $S$ to $p$, it is impossible to determine with certainty the value of $n=|S|$.
\end{lemma}
 \textbf{proof:} Since the set of query points is finite, there must be non-empty areas of the plane
(or bounding box) containing no query points.  Let $d$ be any such disc containing no query point.
Suppose the set $S$ has a site $q_{0}$ in the center of $d$, and sites $q_{1},q_{2},\ldots,q_{6}$
on the boundary of $d$ at intervals of $\pi / 3$. Any query outside of $d$---which by
assumption includes all the query points---will be closer to one of  $q_{1},q_{2},\ldots,q_{6}$
than to $q_{0}$ and thus no query point will discover $q_{0}$.

In fact by spacing $7$ sites around the disc $d$, we can ``hide'' an arbitrary number of sites 
inside of $d$ that will not be discovered by any query. This leads to the following even
stronger lemma.

\begin{lemma}
Given a set $S$ of $n$ planar points 
in a bounding box of area greater than $0$, 
using any finite number of queries $p$ that return only the unique label $q$ of the closest site 
(or sites) in $S$ to $p$, let $N$ be the number of distinct sites returned by the queries.
It is possible that $N/n < \epsilon$ for any positive constant $\epsilon$.
\end{lemma}

That is, the number of sites in $S$ discovered by queries---and included in any cloned
Voronoi diagram---no matter how many queries are used, could be any arbitrary small  fraction 
of the actually number of points in $S$. 
Phrased the other way, we could miss any arbitrary large 
percentage of the sites in $S$. The following corollary also follows directly from this.

\begin{corollary}
Given a set $S$ of $n$ planar points 
in a bounding box of area greater than $0$, 
it is impossible to exactly compute (or clone) $V(S)$ using any finite number of queries $p$ that 
return only the unique label $q$ of the closest site  (or sites) in $S$ to $p$.
\end{corollary}

\subsection{A Quadtree Approximation Approach}
Before looking at how to explicitly construct an approximate Voronoi diagram from queries
of the third sort, we note that we could also use queries to construct a standard quadtree 
that would support $\epsilon$-approximate nearest neighbor searches, and could be used
to construct an $\epsilon$-approximate Voronoi diagram.

The algorithm for constructing a quadtree using queries on the set of unknown sites is
very similar to the standard algorithm for constructing a quadtree for a known set of sites. 
We begin with a bounding box, and query each of its corners, labeling them with the
nearest site from $S$ returned by our query. If all four corners are labeled with the 
same nearest site $q \in S$, then by convexity of Voronoi regions the entire box is inside the 
Voronoi region for $q$---or, alternately, $q$ is the nearest neighbor in $S$ to every point 
in the box. 

If the corners of the bounding box have more than one label, we recursively subdivide it
by making five new queries: one on the center of each side, and one in the center of the rectangle.  
The box now has $9$ labels from queries, which can be used as the $4$ corners of 
$4$ separate boxes. The original box has a pointer to each of its four child boxes: upper left,
upper right, lower left and lower right. Any of these smaller boxes that has 
corners with different labels gets recursively subdivided again into four children.
This process will continue indefinitely. 
We stop when any bounding box has a diagonal
of length less than a fixed constant $\epsilon$.

Unfortunately, the run-time, depth of the quadtree, and number of queries required, are all inversely 
related to $\epsilon$ rather than to the number of points in $S$, and the algorithm can never determine
with certainly the size $n$ of the set $S$. 
The number of probes and the time to construct the tree is $O(N\frac{1}{\epsilon})$, where $N\le n$ is the number of sites discovered by the algorithm.
The depth of the quadtree we construct is $O (\log (\frac{1}{\epsilon}))$. 

\subsection{A Voronoi Diagram Approach}
\fi
We now show that an approximate Voronoi diagram can be constructed to 
answer nearest-neighbor queries, with a probing process
that uses $O(N \log ({1}/{\epsilon}))$ queries and 
$O(N (\log N  + \log (1/\epsilon)))$ time,
where $N\le n$ is the number of discovered
sites in $S$. (Any two sites separated by at least $\epsilon$ will be distinguished and discovered.)
The main idea of the algorithm is to build an approximation to the Voronoi cell of each known site, using $O(\log(1/\epsilon))$ queries per feature of the cell. This sequence of queries either finds a sufficiently accurate approximation for the location of that feature or discovers the existence of another site label.
We begin by querying each corner of our bounding box to find the label of the site in whose region
that corner belongs. For any side of the box whose corners are in different Voronoi regions, 
we do a binary search to find, within a distance of $\epsilon^{2}$,  the edge of the Voronoi region for 
each different site. This may discover new Voronoi regions. For each new region discovered,
we also do a binary search to discover its edges to within $\epsilon^{2}$.  
Each binary search requires  $O (\log\frac{1}{\epsilon})$ queries and time. 
The result is an ordered list of Voronoi edges crossing each side of the bounding
box.

A second similar search a distance of $2 \epsilon$ from each side
of the bounding box will find the same Voronoi edges---or will
discover a new Voronoi region, indicating that the Voronoi edge has
ended.  For those Voronoi edges that have not ended within $2 \epsilon$, 
we compute an approximation of the line containing the
Voronoi edge---that is, the perpendicular bisector of the two sites
whose labels we know. An argument using similar triangles shows that our approximation
of this edge is accurate enough that we can determine to within a
distance of $< \epsilon$ where this edge crosses the far boundary.
We do a doubling search of queries out along each discovered Voronoi edge, 
and then a binary search back once we have moved past the end of the edge, to
find where it ends.  
\ifFull
(For those edges ending within $2 \epsilon$ of
the edge of the box, we use a constant number of queries in a circle of
radius $\epsilon$ around the edge to find where it ends, or to
determine that the edge is too short to be included in our
approximation.)
\fi
Thus, three binary searches of $O (\log (\frac{1}{\epsilon}))$ queries
and time each suffice to discover complete approximations of each
Voronoi edge intersecting the bounding box, including an approximate
location of the Voronoi vertex terminating these edges.  In the
worst case, our approximation is within $\epsilon$.
A constant number of queries in the vicinity of each Voronoi vertex
will discover the other edge or edges coming out of the Vertex. We
repeat this process for each new Voronoi edge as it is discovered,
until every Voronoi edge has both ends terminated at Voronoi vertices,
at which time the approximate Voronoi diagram is complete.

\begin{theorem}
Given a set $S$ of $n$ points in $\R^2$, we can construct a planar
subdivision, $V'$, using $O(N\log\frac{1}{\epsilon})$ Type-3 queries and 
$O(N(\log N+\log\frac{1}{\epsilon}))$ time,
where $N<n$ is the number of discovered
sites in $S$, such that any two sites separated by 
at least $\epsilon$ will be distinguished and discovered and each
point on the 1-dimensional skeleton of $V$ is within distance $\epsilon$ of 
a point on the 1-dimensional skeleton of the Voronoi diagram, $V(S)$, of $S$.
\end{theorem}

\ifFull
\section{Conclusion and Open Problem}
We have given a number of efficient algorithms for cloning Voronoi
diagrams, under various assumptions about the types of queries that
are allowed.
Our methods improve those that would be implied by using existing
polytope probing strategies, with several being based on the use of a novel
retroactive dictionary implementation.
We leave as an open problem whether it is possible to efficiently
clone a Voronoi diagram given only approximate distance queries.
\fi

{
\raggedright
\bibliographystyle{abbrv}
\bibliography{geom,goodrich,clone,k_anonymity}
}

\ifFull
\clearpage
\begin{appendix}
\section{Algorithmic Details for Method with Type-1 Queries}

In this appendix we provide more detail to our algorithm to exactly 
clone a Voronoi diagram using Type-1 queries,
following the de~Berg {\it et al.}~\cite{bkos-cgaa-97} presentation of 
Fortune's sweep-line algorithm~\cite{f-savd-87}, which is based on using a
beach line made up of parabolic arcs instead of hyperbolic arcs 
(readers are referred to de~Berg {\it et al.}~for the complete
details). 

The data structures consist of the Voronoi 
diagram $V(S)$ itself that begins empty 
and is constructed 
during the sweep, a priority queue $Q$ of events (allowing for deletion of arbitrary events), 
and a retroactive dictionary, $T$, 
to keep track of the active sites in the beach line at each time of the plane
sweep. 
For each time value (which is determined by a $y$-coordinate of a point event
or Voronoi vertex event),
$T$ is searchable by $x$-coordinate at any time for which the beachline
exists.
Our algorithm makes queries into a 
\emph{black box} that has complete
information about $S$ and represents the actions of the
data owner, Alice.  The algorithm follows:

\begin{itemize}
\item \textbf{Initialization}:  Using as initial queries the top left and top right corners of the bounding box, $p_{l}$ and $p_{r}$ respectively, determine the nearest neighbors in $S$ to $p_{l}$ and $p_{r}$. 
Let the results of these queries be sites $q_{l}$ and $q_{r}$ respectively. 
If  $q_{l}=q_{r}$---that is, both corners $p_{l}$ and $p_{r}$ have the same nearest neighbor 
$q=q_{l}=q_{r} \in S$---then only one Voronoi region $V(q)$ (and no Voronoi edges) 
intersects the top of the bounding box. Add site $q$ to the event queue $Q$, terminate
the \textbf{Initialization} stage, and go to the \textbf{Plane Sweep} stage of the algorithm. 

\textbf{Initialization Recursive Step}: If $q_{l} \neq q_{r}$---that is the two corners are in different Voronoi 
regions---then we recursively subdivide as follows.  Put sites $q_{l}$ and $q_{r}$ in $Q$.  
Then compute (in $O(1)$ time) the perpendicular bisector $E$ of the segment $q_{l}q_{r}$ 
on which lies the Voronoi edge between regions $V(q_{l})$ and $V(q_{r})$, \emph{if} these
regions share an edge. Let $p_{m}$ be the intersection of $E$ with the top of the bounding
box. (This intersection must exist, since the two corners are in different Voronoi regions.)  
Query $p_{m}$. If $q_{l}$ and $q_{r}$ are the nearest neighbors of $p_{m}$, then we are done. 
If not, let $q_{m}$ (with $q_{m} \neq q_{l}$ and $q_{m} \neq q_{r}$) be the nearest site in 
$S$ to $p_{m}$. Add $q_{m}$ to $Q$, and make two recursive calls to this step, one with 
$p_{l}$ and $p_{m}$ taking the roles of $p_{l}$ and $p_{r}$ above, and one with 
$p_{m}$ and $p_{r}$ taking the roles of $p_{l}$ and $p_{r}$.

At the end of the initialization, every site whose Voronoi region intersects the top
of the bounding box will be in the event queue. We now carry out the following plane 
sweep using the event queue.

\item \textbf{Plane Sweep}
 \begin{algorithmic}[100]
\WHILE{the event queue $Q$ is not empty}
\STATE Remove the event with the largest $y$-coordinate.
\STATE It is either a SiteEvent $p_{i}$  or a TentativeVoronoiVertex $q_{i}$.  
\IF{the event is a site event}
		\STATE HandleSiteEvent($p_{i}$) 
\ELSE
		\STATE HandleTentativeVoronoiVertex($q_{i}$)  
\ENDIF
\ENDWHILE
\end{algorithmic}

The difference between this and the standard Fortune plane sweep is how the events are
handled, as described below.

\item \textbf{HandleSiteEvent($p_{i}$)}:
\begin{algorithmic}[100]
\STATE Search in $T$ for the site $p_{j}$ whose region contains the $x$-coordinate of $p_{i}$
when the sweep-line is at the $y$-coordinate of $p_{i}$. (Since $p_{j}$ is already in $T$, it
must have a $y$-coordinate larger than that of $p_{i}$.)
\STATE To verify that we are ready to process $p_{i}$---that its tentative Voronoi edge directly
above it is an actual Voronoi edge, and we did not miss any site between $p_{i}$ and 
$p_{j}$---compute the point $q$ on the bisector of $p_{i}p_{j}$ that has the
same $x$-coordinate as $p_{i}$. This would be a point on the Voronoi edge between $p_{i}$ and
$p_{j}$ if that edge exists.  
\STATE Query point $q$ to find its nearest neighbor(s) in $S$. 
\IF{the two nearest neighbors of $q$ are $p_{i}$ and $p_{j}$}
		\STATE Insert $p_{i}$ into the retroactive dictionary, $T$.
		\STATE Add the verified edge $p_{i}p_{j}$ to $V(S)$
		\STATE Complete processing this as a 
		site event following the details of \cite{bkos-cgaa-97}. This includes c
		constructing tentative Voronoi vertices from the edge $p_{i}p_{j}$ and its 
		neighbors, and adding them to $Q$ if not already present. 
		\STATE Delete $p_{i}$ from $Q$. 
\ELSE
		\STATE Leave $p_{i}$ in $Q$
		\STATE Let $p_{k}$ be the nearest neighbor of $q$.  
		\STATE Add $p_{k}$ to $Q$.   We then continue processing. Note that if $p_{k}$ has
		a higher $y$-coordinate than $p_{i}$, then we have to backtrack our sweep line.
\ENDIF
\end{algorithmic}

We next describe how to handle a tentative Voronoi vertex (see
Fig.~\ref{fig:event2}).

\begin{figure}[hbt]
\begin{center}
\includegraphics[width=3.5in]{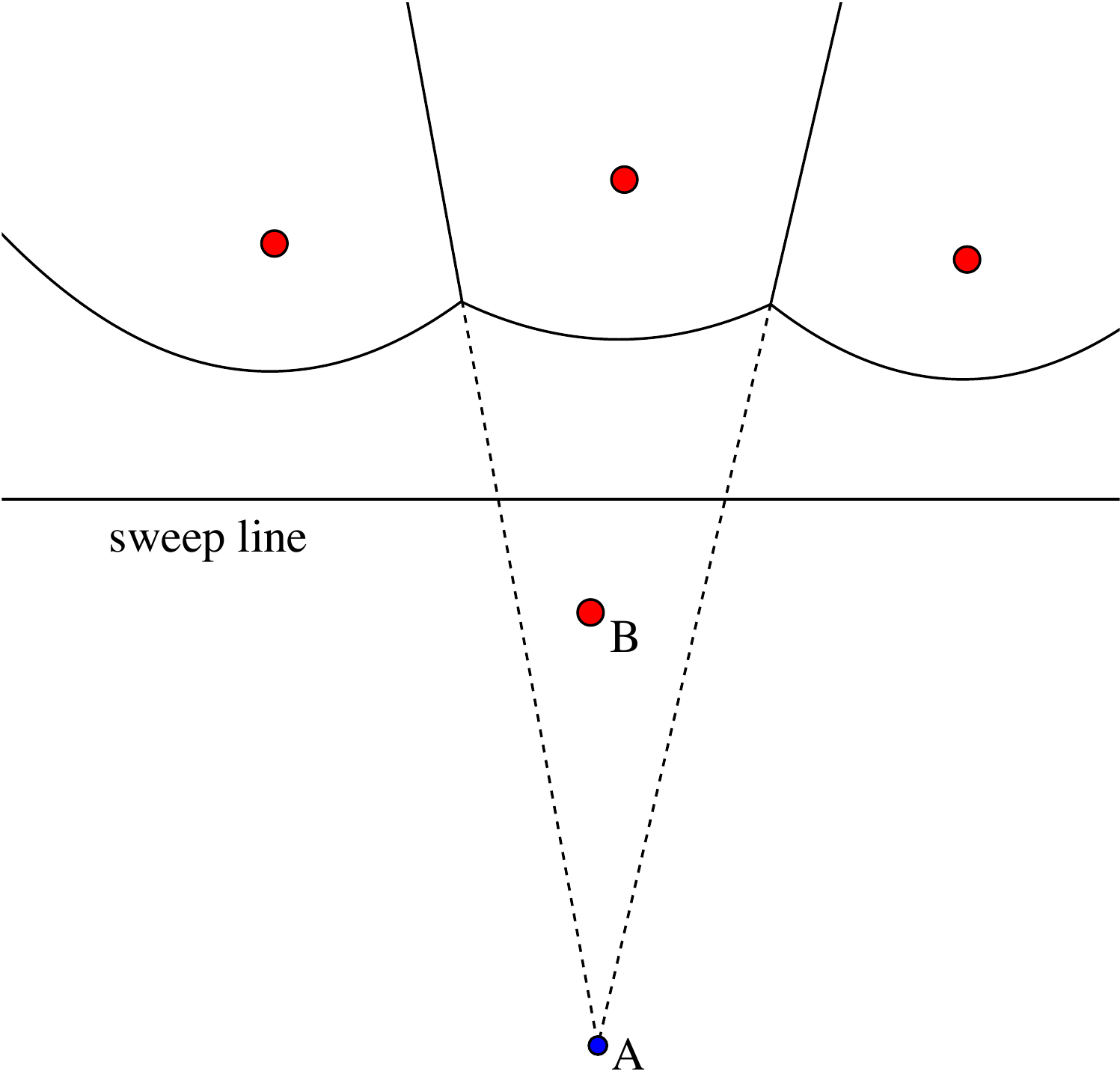}
\caption{A Voronoi vertex event (A) that discovers a new site (B).}
\label{fig:event2}
\end{center}
\end{figure}

\item \textbf{HandleTentativeVoronoiVertex($q_{i}$)}:
\begin{algorithmic}[100]
\STATE Query point $q_{i}$ to determine if its three actual nearest neighbors are the
three sites whose Voronoi regions generated this tentative Voronoi vertex.
\IF{the $q_{i}$ has the expected three nearest neighbors}
		\STATE Perform a deletion in $T$ 
		for beachline segment that ends at this vertex
		at this time in the sweep 
		\STATE Add to $V(S)$ the Voronoi vertex $q_{i}$ and the edges above it.
		\STATE Delete $q_{i}$ from $Q$. 
		\STATE Complete processing this as a 
		CircleEvent following the details of \cite{bkos-cgaa-97}.
\ELSE
		\STATE Delete $q_{i}$ from $Q$.  
		\STATE Add to $Q$ all of the sites that are nearest neighbors of query point $q_{i}$ 
		and continue.
\ENDIF
\end{algorithmic}

\end{itemize}

\end{appendix}
\fi

\end{document}